\documentclass{article}
\usepackage[T1]{fontenc}

\usepackage{geometry}
\usepackage[english]{babel}
\usepackage[dvipsnames]{xcolor}
\usepackage{amssymb}
\usepackage{amsmath}
\usepackage{amsfonts}
\usepackage{amsthm}
\usepackage{appendix}
\usepackage{comment}
\usepackage{tikz} 
\usepackage{subcaption}
\usepackage{authblk}
\usepackage{latexsym} %Des symboles de maths en +
\usepackage{xspace}

\usepackage[breaklinks,hidelinks]{hyperref} 
\usepackage{xcolor} 
\usepackage[]{knowledge} 
\usepackage{enumerate}
\usepackage[shortlabels]{enumitem}

\newtheorem{theorem}{Theorem}[section]

\newtheorem{corollary}[theorem]{Corollary}

\newtheorem{lemma}[theorem]{Lemma}
\newtheorem{definition}[theorem]{Definition}
\newtheorem{claim}{Claim}[]
\newtheorem{conjecture}[theorem]{Conjecture}

\knowledgeconfigure{notion}
\knowledgeconfigure{quotation}

\title{Zero Forcing Sets in Temporal Graphs\thanks{This work is an initiative from the Journées Graphes et Algorithmes 2025. This research was made possible through the French government IDEX-ISITE initiative 16-IDEX-0001 (CAP 20-25), International Research Center ``Innovation Transportation and Production Systems'' of the I-SITE CAP 20-25 and grant ANR-22-CE48-0001 (projet TEMPOGRAL). }}

\author[1]{Julien Baste}
\author[2]{Simon Dreyer}
\author[3]{Clara Marcille}
\author[4]{Mikaël Rabie}
\author{Ronan Toullec-Streicher}

\affil[1]{Univ. Lille, CNRS, Centrale Lille, UMR 9189 CRIStAL, F-59000 Lille, France}
\affil[2]{LIRMM, Université de Montpellier, CNRS, Montpellier, France }
\affil[3]{Université Clermont Auvergne, LIMOS, Clermont-Ferrand, France}
\affil[4]{IRIF, CNRS and Université Paris Cité, France}

\date{}

\definecolor{darkcyan}{rgb}{0, 0.55, 0.55}
\definecolor{darkgreen}{rgb}{0, 100, 0}

\newcommand{\x}[1]{x_{#1}}
\newcommand{\nx}[1]{\neg{x_{#1}}}
\newcommand{\vx}[1]{v_{#1}}
\newcommand{\nvx}[1]{\overline{v_{#1}}}
\newcommand{\C}[1]{u_{#1}}

\newcommand{\entry}{\texttt{entry}\xspace}

\newcommand{\probl}[3]{
\begin{flushleft}
\fbox{
\begin{minipage}{\textwidth}
\noindent {\sc #1}\\
          {\bf Input:} #2\\
          {\bf Output:} #3
\end{minipage}}
\medskip
\end{flushleft}
}

\begin{document}

\knowledgedirective{discrete}[black]{color={#1},emphasize,md}

\knowledge{notion}
	| Zero Forcing process
 	| Zero Forcing problem
 	| Zero Forcing@process

\knowledge{notion}
 | zero forcing sets@static
 | zero forcing set@static
 | zero forcing process@static

\knowledge{notion}
 |	 zero forcing number@static
 | size of a zero forcing set

\knowledge{notion}
|	 corrupted

\knowledge{notion}
	| corrupting function
	| temporal@rule
	| corrupts

\knowledge{notion}
	| non-corrupted
	| not corrupted

\knowledge{notion}
	| color-change rule
	| corruption rule
	| propagation
	| process@corruption
	| static@rule
	| corruption

\knowledge{notion}
	| temporal zero forcing set

\knowledge{notion}
	| partial witness

\knowledge{notion}
	| witness

\knowledge{notion}
	| Temporal Zero Forcing Set@problem
	| temporal problem@problem
	| temporal equivalent@problem
	| temporal zero forcing problem

\knowledge{notion}
	| Zero Forcing Set@problem
	| problem@static

\knowledge{notion}
	| temporal zero forcing number
	| TZ

\knowledge{notion}
	 |Multicolored Clique@problem

\knowledge{notion}
	| underlying graph
 | Underlying graph

\knowledge{notion}
	 |lifespan

\knowledge{color=black}
	| temporal graph
 	| temporal graphs

\knowledge{notion}
	| 1-step forcing set

\maketitle

\begin{abstract}
The Zero Forcing (or corruption) of a graph is the problem of finding a minimum-size ``corrupting'' set. It corresponds to a subset of its vertices that can corrupt the whole graph by iterating the following rule: if a corrupted vertex has exactly one neighbor that is not yet corrupted, the neighbor gets corrupted. The iteration of this process comes from the fact that the corruption of a vertex might enable new corruptions (from itself or some of its neighbors). For this reason, one can consider a step of corruption, where all the possible instances of the corruption rule are applied at once.

This paper investigates Zero Forcing on temporal graphs, where the topology of the graph evolves throughout the experiment. At each time step (or snapshot) of the graph, a step of corruption is resolved wherever possible. We study the problem of finding a minimum-size corrupting set such that the whole (temporal) graph is corrupted at the end of the experiment. We present a panorama of results, including NP-hardness in some not-so-restrictive scenarios, polynomial algorithms, and a solution to an open question when the whole graph must be corrupted in a single step. 
\end{abstract}

\section{Introduction}

Spreading information, or corruption, in a network is a way to study how information can move through entities. In particular, spreading may not be easy and can depend on restrictions on the entity spreading the information. Can an agent spread information to all its neighbors at once? Or one by one? Or does it need all its neighbors but one to already have it to spread information to the last one? Moreover, networks might change over time. This is the case in social networks where connections can evolve, or when moving agents make contact when they are close enough to each other. One can then ask, given a (dynamic) network, how many entities need the initial information for it to spread, and how fast.

In this article, we consider the ""Zero Forcing process"", which is a dynamic process to spread some corruption. Initially, some of the elements are "corrupted". At each time step, if all but one of a vertex's neighbors are "corrupted", the last neighbor becomes "corrupted". More formally, a graph $G$ starts with a subset of its vertices being \textit{"corrupted"}, and the corruption spreads through the following ""process@@corruption"": at each step, if a "corrupted" vertex $u$ has exactly one "non-corrupted" neighbor $v$, then $v$ becomes "corrupted", and so on until there is no possible "corruption" is possible, or until the whole graph is "corrupted". The goal of "Zero Forcing@@process" is to find the minimum size subset of vertices of an input graph such that, by applying the "corruption rule" as many times as necessary, the whole graph is eventually "corrupted". This process was introduced independently in the field of linear algebra~\cite{van2008zero} to bound the minimum rank of a graph, and in the field of quantum systems~\cite{PhysRevA.79.060305,PhysRevLett.99.100501}. Because of its applications in numerous areas, and in particular in linear algebra, the problem has drawn significant attention over the years. We refer the reader to Chapter 9 of~\cite{hogben2022inverse} for a survey of this notion. Even more recently, some connection has been drawn between the "Zero Forcing problem" and the celebrated Burning Tree problem~\cite{abiad2026propagation}. 

%The ""Zero Forcing process"" is a dynamic process on a graph to spread some corruption. A graph $G$ starts with a subset of its vertices said to be \textit{"corrupted"}, and the corruption spreads through the following ""process@@corruption"": at each step, if a corrupted vertex $u$ has exactly one "non-corrupted" neighbor $v$, then $v$ becomes "corrupted", and so on until there is no possible corruption is possible, or until the whole graph is "corrupted". The goal of "Zero Forcing@@process" is to find the minimum size subset of vertices of an input graph such that, by applying the "corruption rule" as many times as necessary, the whole graph is eventually "corrupted". This process was introduced independently in the field of linear algebra~\cite{van2008zero} to bound the minimum rank of a graph, and in the field of quantum systems~\cite{PhysRevA.79.060305,PhysRevLett.99.100501}. Because of its applications in numerous areas, and in particular in linear algebra, the problem has drawn significant attention over the years. We refer the reader to Chapter 9 of~\cite{hogben2022inverse} for a survey of this notion. Even more recently, some connection has been drawn between the "Zero Forcing problem" and the celebrated Burning Tree problem~\cite{abiad2026propagation}. 

A quite recent addition to the theory of "Zero Forcing@@process" is the notion of time in the dynamic process of "propagation". In~\cite{aazami2008hardness}, the author initiates the study of a time-constrained version of "zero forcing sets@@static". In this variant, the main focus is the "propagation" speed, that is, the number of steps required to fully corrupt a graph, assuming all currently possible applications of the "color-change rule" occur simultaneously. The minimum sum of the "size of a zero forcing set" and its "propagation" time is known as the \textit{throttling number}, and has since then been thoroughly studied on static graphs~\cite{aazami2008hardness,butler2013throttling,carlson2018throttling,hogben2023newstructuresapplicationsvariants}. In this work, we add a layer of temporality by considering topological variations of the graph, introducing the "Zero Forcing problem" on "temporal graphs", where at each time, the set of edges can change. This kind of question already exists in~\cite{9086120}, where the authors investigate topological perturbations of the graph under some restricted set of graph operations. 

Using the framework of "temporal graphs", we generalize all these operations and provide a variety of complexity results on this new variant of the "Zero Forcing problem". We prove that even in the case where the "underlying graph" (i.e., when we consider all the edges that appear at least once through time) is as simple as a star, the problem is NP-hard, and even $W[1]$-hard when parametrized by the size of the forcing set. The problem remains hard even when we bound each vertex's degree and the total time. On the other hand, we provide polynomial algorithms for bounded-degree temporal trees, as well as stars with a polylogarithmic number of time steps. For the single time-step case, which directly relates to the classical "Zero Forcing problem", we provide the first NP-hardness result, solving an open question about $1$-step corruption~\cite{aazami2008hardness,butler2013throttling}.

%For the latter, we give a polynomial-time algorithm to solve the "temporal zero forcing problem". We complete our study by answering different questions from the literature on static "zero forcing process@@static", namely $1$-step corruption~\cite{aazami2008hardness,butler2013throttling}.
This paper starts with some preliminaries in Section \ref{sec:prelim}. We provide some \texttt{NP}-hardness results in Section~\ref{sec:main-hard}. Section \ref{sec:lifespan1} investigates a case where the "corruption" has to be done in a unique step. In Section~\ref{sec:temporal-stars}, we explore different cases where the "underlying graph" is a star, and Section~\ref{sec:trees} is about the case where the "underlying graph" is a tree. We conclude the paper in Section~\ref{sec:ccl}.

\section{Preliminaries}\label{sec:prelim}

In this work, we use the following notations for basic graph notions.
We denote by $\mathbb{N}$ the set of non-negative integers.
Given two integers $a$ and $b$, we denote by $[a,b]$ the set $\{i \in \mathbb{N} \mid a \leq i \leq b\}$.
A graph $G$ is a pair $(V,E)$ where $V$ is a set and $E \subseteq V \times V$.
$V$ is the set of \emph{vertices} and $E$ is the set of \emph{edges}.
Given a vertex $v \in V$, we denote by $N_G(v) = \{v' \mid \{v,v'\} \in E\}$ the open neighborhood of $v$ in $G$, and
$N_G[v] = N_G(v) \cup \{v\}$ the close neighborhood of $v$. 
We say a set $S\subset V$ is a ""zero forcing set@@static"" if, for all $u\in V$, either $u\in S$, or $u$ can be reached using the so-called ``propagation rule'', which is as follows. Vertices in $S$ are said to be \textit{"corrupted"}, and any vertex that is the only "non-corrupted" neighbor of a "corrupted" vertex is "corrupted" as well, and the "propagation" rule is applied iteratively.
The ""zero forcing number@@static"" of a graph $G=(V, E)$ is the minimum size of a "zero forcing set@@static".

\begin{definition}
  A ""temporal graph"" is a triplet $\mathcal{G}=(V,E,\lambda)$ where  $V$ is the set of vertices, $E \subseteq {2 \choose V}$ is the set of edges and $\lambda: E \to 2^{\mathbb{N}}$.
  The ""lifespan"" of $\mathcal{G}$, denoted $T(\mathcal{G})$, is the last time at which there is an edge, i.e., $T(\mathcal{G})=\max\bigcup_{e\in E}\lambda(e)$.
  For $i \in \mathbb{N}$, the snapshot $i$ of $\mathcal{G}$ is $G_i(\mathcal{G})=(V, E_i)$ where $E_i=\{e\in E|i\in\lambda(e)\}$ and $E= \cup_{0\leq i \leq T(\mathcal{G})} E_i$.

  The ""underlying graph"" of $\mathcal{G}$ is the graph without its temporality, i.e., $(V,E)$.
\end{definition}
Note that in the definition of a "temporal graph", the $\lambda$ function describes the temporality of each edge, i.e., the timestamps at which the edge is present.
  For any edge $e\in E$, $\lambda(e)$ denotes the times at which edge $e$ appears.
When $\mathcal{G}$ is clear from the context, we write, for each $i \in \mathbb{N}$, $G_i$ instead of $G_i(\mathcal{G})$, and $T$ instead of $T(\mathcal{G})$.

In this work, we consider specific subclasses of "temporal graphs", based on their "underlying graph". Throughout the paper, whenever we refer to some structural property on a "temporal graph" (that is, being a clique, being bipartite, etc.), we refer to its "underlying graph" having this property. For instance, a "temporal graph" $\mathcal{G}=(V,E,\lambda)$ is a temporal clique if, for all pairs of distinct vertices $u, v \in V$, we have $uv\in E$ and $\lambda(uv)$ is non-zero. In particular, we do not care if some snapshot of $\mathcal{G}$ is not a clique.

\begin{definition}
  \label{def:corr}
  Given a "temporal graph" $\mathcal{G}=(V,E,\lambda)$ with "lifespan" $T$, we say that a function $\texttt{corr} : [0,T] \to 2^V$ is a ""corrupting function"" if 
  $\texttt{corr}(T) = V$ and 
  for each $i \in [1,T]$, 
  $v \in \texttt{corr}(i)$ if and only if $v \in \texttt{corr}(i-1)$ or there exists $v' \in \texttt{corr}(i-1)$ such that
  $N_{G_i}(v') = N \cup \{v\}$ with $N \subseteq \texttt{corr}(i-1)$,
  or such that $ N_{G_i}(v') \setminus \texttt{corr}(i-1) = \{v\}$ (i.e., $v$ is the only neighbor of $v'$ not "corrupted" in snapshot $i-1$).

  A set $S$ is a ""temporal zero forcing set"" of $\mathcal{G}$ if there exists a "corrupting function" $\texttt{corr}$ of $\mathcal{G}$ such that $S = \texttt{corr}(0)$.

We call the ""temporal zero forcing number"" of $\mathcal{G}$, noted $"TZ"(\mathcal{G})$, the minimum cardinality of a "temporal zero forcing set" of $\mathcal{G}$.
\end{definition}

In this article, given a "temporal graph" $\mathcal{G}=(V,E,\lambda)$, % on a vertex set $v$, and a corrupting function $\texttt{corr}$, we say that a vertex $v \in V$
we say that at a given time $i \in [1,T]$,  a vertex of $v$ is ""corrupted"" if $v \in \texttt{corr}(i)$ and ""not corrupted"" otherwise.
% Moreover a vertex is \emph{corrupted} at time $i \in [1,T]$, if $v \in \texttt{corr}(i) \setminus \texttt{corr}(i-1)$.
Given $i \in [1,T]$, if $v' \in \texttt{corr}(i-1)$ is such that
  $N_{G_i}(v') = N \cup \{v\}$ with $N \subseteq \texttt{corr}(i-1)$ for some $v \not \in \texttt{corr}(i-1)$, then we say that $v'$ "corrupts" $v$ at time $i$ and $v$ is "corrupted" after time $i$. 
  Intuitively, if a "corrupted" vertex has only one neighbor that is "not corrupted" at a given time, then this neighbor becomes "corrupted" after this step. 

We recall the static \textsc{"Zero Forcing Set@@problem"}, which is defined as follows:
\probl{""Zero Forcing Set@@problem""}
{A graph  $G=(V,E)$ and an integer $k$.}
{A "zero forcing set@@static" $S$ of $G$ of size at most $k$ or a correct output that such a set does not exist.}

We are now ready to introduce the temporal equivalent of the \textsc{"Zero Forcing Set@@problem"} problem.

\probl{""Temporal Zero Forcing Set@@problem""}
{A "temporal graph"  $\mathcal{G}=(V,E,\lambda)$ and an integer $k$.}
{A "temporal zero forcing set" $S$ of $\mathcal{G}$ of size at most $k$ or a correct output that such a set does not exist.}

Observe that while the statement of \textsc{"Temporal Zero Forcing Set@@problem"} seems very close to the original \textsc{"Zero Forcing Set@@problem"}, we think it is important to underline that the two corruption rules ("static@@rule" and "temporal@@rule") behave differently with regard to the dynamicity, in the sense that whenever a "corruption" happens in the temporal setting, any new "corruption" that becomes enabled will have to wait for a later snapshot.
However, one of the early results on \textsc{"Zero Forcing Set@@problem"} by Aazami, who proves the \texttt{NP}-hardness of the problem in~\cite{aazami2008hardness}, translates directly to \textsc{"Temporal Zero Forcing Set@@problem"}. Indeed, consider a static graph $G=(V, E)$ as a "temporal graph" $\mathcal{G}=(V,E,\lambda)$ where $\lambda$ is such that every snapshot of $\mathcal{G}$ is $G$ with "lifespan" $|V|$, and observe that in this case, \textsc{"Temporal Zero Forcing Set@@problem"} falls back in the statement of \textsc{"Zero Forcing Set@@problem"}. Intuitively, this leads us to think that the "temporal problem@@problem" is in a sense harder than its static counterpart. Our work confirms this intuition, and we investigate in detail the newly added layers of complexity that come with this generalization. It can also be observed that given a "temporal graph" $\mathcal{G}=(V,E,\lambda)$ and some set $S\subset V$, it is easy to check whether $S$ is a "temporal zero forcing set" of $\mathcal{G}$, as it is sufficient to emulate the "corruption" process for each snapshot, applying all possible "corruption" at each step. 

\section{Hardness of the Temporal Zero Forcing Set problem}\label{sec:main-hard}

We first give ground-setting results on \textsc{"Temporal Zero Forcing Set@@problem"}, starting with an NP-hardness result in the case where the "underlying graph" is a star. Although the result in Theorem~\ref{th:np-hard} is weaker and, on numerous accounts, implied by later results in this work, we take time to present its arguments, as they lay the groundwork for the following proofs. 
\begin{theorem}\label{th:np-hard}
  \textsc{"Temporal Zero Forcing Set@@problem"} is \texttt{NP}-hard when
  restricted to "temporal graphs" whose "underlying graph" is a star.
\end{theorem}

Before proceeding, we refer the reader to the introduction of~\cite{TOVEY198485} for the definition of \textsc{$3$-SAT}. We prove Theorem~\ref{th:np-hard} by a simple reduction from \textsc{$3$-SAT}. Since the problem directly lies in \texttt{NP}, we only have to prove the following lemma.

\begin{lemma}\label{lem:np_hard_star}
    Given a \textsc{$3$-SAT} formula $\Phi$ with $k$ variables, there exists a temporal star $\mathcal{G}$ such that $"TZ"(\mathcal{G})= k+1$ if and only if $\Phi$ is satisfiable, and $\mathcal{G}$ can be computed in time polynomial in $|\Phi|$.
\end{lemma}

Let us first explain the construction of the "temporal graph" $\mathcal{G}$ from a \textsc{$3$-SAT} formula $\Phi$. A representation of this construction is given in Figure~\ref{fig:star3sat}. 
\begin{figure}[t]
\begin{subfigure}{0.31\textwidth}
        \centering
\scalebox{0.9}{
\begin{tikzpicture}
   \tikzstyle{circlenode}=[draw,circle,minimum size=70pt,inner sep=0pt]
    \tikzstyle{whitenode}=[draw,circle,fill=white,minimum size=18pt,inner sep=0pt]
    \tikzstyle{rednode}=[draw,circle=red,fill=red,minimum size=12pt,inner sep=0pt]
    \tikzstyle{bluenode}=[draw,circle=blue,fill=blue,minimum size=12pt,inner sep=0pt]
    \tikzstyle{greynode}=[draw,circle=purple,fill=black!50,minimum size=12pt,inner sep=0pt]
    \tikzstyle{nonode}=[draw=white,circle=red,fill=white,minimum size=12pt,inner sep=0pt]
 
\draw (0,0) node[whitenode] (a) {$c$};
%\draw (-2,-2) node[whitenode] (b1) {75};
\draw (-2,0) node[whitenode] (b2) {$c'$};
\draw (-1.5,1.5) node[whitenode] (b3) {$\vx{i}$};
\draw (0,-2) node[whitenode] (b4) {$\C{j,2}$};
\draw (0,2) node[whitenode] (b5) {$\nvx{i}$};
\draw (1.5,-1.5) node[whitenode] (b6) {$\C{j,1}$};
\draw (-1.5,-1.5) node[whitenode] (b9) {$\C{j,3}$};
\draw (2,0) node[whitenode] (b7) {$\vx{i}'$};
\draw (1.5,1.5) node[whitenode] (b8) {$\nvx{i}'$};

%\draw (a) edge node {} (b1);
\draw (a) edge node {} (b2);
\draw (a) edge node {} (b3);
\draw (a) edge node {} (b4);
\draw (a) edge node {} (b5);
\draw (a) edge node{} (b6);
\draw (a) edge node {} (b7);
\draw (a) edge node {} (b8);
\draw (a) edge node {} (b9);
%\draw (0.5,-1.5) edge [bend left, dotted] node {} (-1.5,-0.5);
\draw (-1.2,0.9) edge [bend right, dotted] node {} (-1.5,0.2);
\draw (1.2,-0.9) edge [bend right, dotted] node {} (1.5,-0.2);
\draw (-1.2,-0.9) edge [bend left, dotted] node {} (-1.5,-0.2);

\end{tikzpicture}
}
\caption{"Underlying graph" of $\mathcal{G}$}
        \end{subfigure}
        ~
\begin{subfigure}{0.31\textwidth}
        \centering
\scalebox{0.9}{
\begin{tikzpicture}
   \tikzstyle{circlenode}=[draw,circle,minimum size=70pt,inner sep=0pt]
    \tikzstyle{whitenode}=[draw,circle,fill=white,minimum size=18pt,inner sep=0pt]
    \tikzstyle{rednode}=[draw,circle=red,fill=red,minimum size=12pt,inner sep=0pt]
    \tikzstyle{bluenode}=[draw,circle=blue,fill=blue,minimum size=12pt,inner sep=0pt]
    \tikzstyle{greynode}=[draw,circle=purple,fill=black!50,minimum size=12pt,inner sep=0pt]
    \tikzstyle{nonode}=[draw=white,circle=red,fill=white,minimum size=12pt,inner sep=0pt]
 
\draw (0,0) node[whitenode] (a) {$c$};
%\draw (-2,-2) node[whitenode] (b1) {75};
\draw (-2,0) node[whitenode] (b2) {$c'$};
\draw (-1.5,1.5) node[whitenode] (b3) {$\vx{i}$};
\draw (0,-2) node[whitenode] (b4) {$\C{j,2}$};
\draw (0,2) node[whitenode] (b5) {$\nvx{i}$};
\draw (1.5,-1.5) node[whitenode] (b6) {$\C{j,1}$};
\draw (-1.5,-1.5) node[whitenode] (b9) {$\C{j,3}$};
\draw (2,0) node[whitenode] (b7) {$\vx{i}'$};
\draw (1.5,1.5) node[whitenode] (b8) {$\nvx{i}'$};

%\draw (a) edge node {} (b1);
\draw (a) edge node {} (b2);

\end{tikzpicture}
}
\caption{$G_1=(V,E_1)$}
        \end{subfigure}
        ~
\begin{subfigure}{0.31\textwidth}
        \centering
\scalebox{0.9}{
\begin{tikzpicture}
   \tikzstyle{circlenode}=[draw,circle,minimum size=70pt,inner sep=0pt]
    \tikzstyle{whitenode}=[draw,circle,fill=white,minimum size=18pt,inner sep=0pt]
    \tikzstyle{rednode}=[draw,circle=red,fill=red,minimum size=12pt,inner sep=0pt]
    \tikzstyle{bluenode}=[draw,circle=blue,fill=blue,minimum size=12pt,inner sep=0pt]
    \tikzstyle{greynode}=[draw,circle=purple,fill=black!50,minimum size=12pt,inner sep=0pt]
    \tikzstyle{nonode}=[draw=white,circle=red,fill=white,minimum size=12pt,inner sep=0pt]
 
\draw (0,0) node[whitenode] (a) {$c$};
%\draw (-2,-2) node[whitenode] (b1) {75};
\draw (-2,0) node[whitenode] (b2) {$c'$};
\draw (-1.5,1.5) node[whitenode] (b3) {$\vx{1}$};
\draw (0,-2) node[whitenode] (b4) {$\C{j,2}$};
\draw (0,2) node[whitenode] (b5) {$\nvx{1}$};
\draw (1.5,-1.5) node[whitenode] (b6) {$\C{j,1}$};
\draw (-1.5,-1.5) node[whitenode] (b9) {$\C{j,3}$};
\draw (2,0) node[whitenode] (b7) {$\vx{1}'$};
\draw (1.5,1.5) node[whitenode] (b8) {$\nvx{1}'$};

%\draw (a) edge node {} (b1);
\draw (a) edge node {} (b3);
\draw (a) edge node {} (b7);
%\draw (0.5,-1.5) edge [bend left, dotted] node {} (-1.5,-0.5);

\end{tikzpicture}
}
\caption{$G_2=(V,E_2)$}
        \end{subfigure}
\caption{The temporal star presented in the proof of Theorem~\ref{th:np-hard}, with its static expansion and two first snapshots}\label{fig:star3sat}
\end{figure}
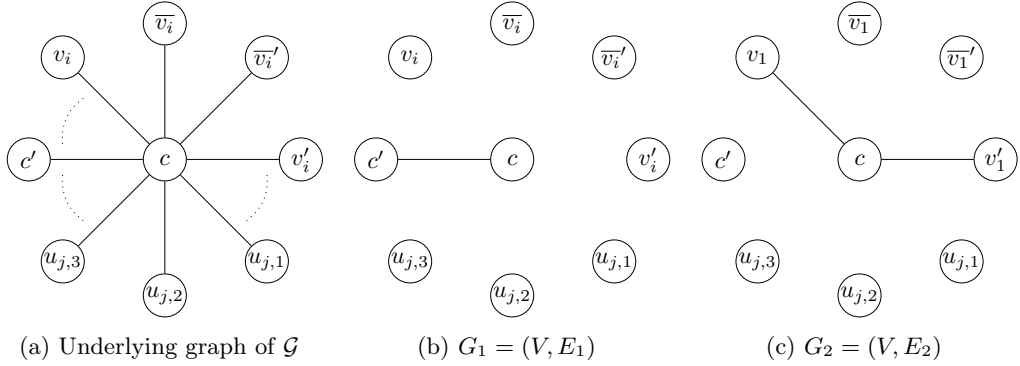

    Consider a \textsc{$3$-SAT} formula $\Phi$ in its conjunctive normal form with exactly $3$ literals per clause. Let $(C_i)_{i\in[1,m]}$ be its clauses. For $i\in[1,m]$, let $(l_{i,j})_{j\in \{1,2,3\}}$ be the literals of $C_i$. %We have $\displaystyle \Phi=\bigvee_{i=1}^m C_i=\bigvee_{i=1}^m \bigwedge_{j=1}^3 l_{i,j}$. 
    Let $k$ be the number of variables used in the formula $\Phi$. For each $(i,j)\in [1,n]\times \{1,2,3\}$, there exists $l \in [1,k]$ such that $l_{i,j}=\x{l}$ or $l_{i,j}=\nx{l}$.

    %Let us construct a temporal star $\mathcal{G}$ such that $"TZ"(\mathcal{G})= k+1$ if and only if $\Phi$ is satisfiable. Figure~\ref{fig:star3sat} illustrates the "underlying graph" of $\mathcal{G}$, as well as its first two snapshots. 
    We start the construction of $\mathcal{G}$ with a single vertex $c$, which we call the \emph{center}. All the other vertices are adjacent to $c$. They are:
\begin{itemize}
\item a vertex $c'$;
\item for each variable $\x{i}$ (with $i \in [1, k]$), four vertices $\vx{i}$, $\nvx{i}$, $\vx{i}'$, and $\nvx{i}'$;
\item for each clause $C_i$ (with $i \in [1,m]$), three vertices $\C{i,1}$, $\C{i,2}$, and $\C{i,3}$.
\end{itemize}

Since $\mathcal{G}$ is a star with center $c$, all other vertices are adjacent to $c$. To describe $\lambda$, we actually describe each snapshot. As all edges are connected to $c$, we describe a snapshot by giving the set of vertices whose (only) incident edge is present at that time:
%, as well as some intuition behind the choice of this set:

\begin{enumerate}
\item snapshot $1$ - $\{c'\}$;
\item snapshots $2$ to $3k+1$ - For each $i\in[1,k]$, we use three consecutive snapshots:
\begin{itemize}
\item at $3i-1$: $\{\vx{i},\vx{i}'\}$;
\item at $3i$: $\{\nvx{i},\nvx{i}'\}$;
\item at $3i+1$: $\{\vx{i}',\nvx{i}'\}$.
\end{itemize}

\item Snapshots $3k+2$ to $3k + 6m + 1$- For each $i\in[1,m]$, we use six consecutive snapshots:
\begin{itemize}
\item at $3k-2 +6i +j$ for each $1\leq j\leq 3$, and $1\leq i\leq m$, if $l_{i,j}=\x{t}$ (resp. $\nx{t}$) for some $t$, we have the set $\{\C{i,j},\vx{t}\}$, (resp. the set $\{\C{i,j},\nvx{t}\}$). 
\item at $3k+1+6i+j$ for each $1\leq j\leq 3$, we have $\{\C{i,j},\C{i,(j \mod 3)+1}\}$. 
\end{itemize}
\item Snapshots $3k + 6m + 2$ to $4k + 6m + 1$ - For each $i\in[1,k]$, snapshot $3k + 6m + 1 + i$ consists in the set $\{\vx{i},\nvx{i}\}$.
\end{enumerate}

While technically $\mathcal{G}$ should be expressed as a function of $\Phi$, we make the choice of simplifying the notation to make the proof easier to read.
\begin{lemma}
  $\mathcal{G}$ has a size polynomial in $|\Phi|$.
\end{lemma}
\begin{proof}
  The "temporal graph" $\mathcal{G}$ corresponds to the sequence of its snapshots, where each snapshot is a graph on at most $4k+2$ vertices, and there are $4k + 6m + 1$ snapshots.
\end{proof}
The main tool of this proof is that a vertex can only become "corrupted" when its incident edge appears in a snapshot.

    \begin{lemma}\label{cl:3sat-c}
        Any "temporal zero forcing set" of $\mathcal{G}$ contains either $c$ or $c'$ and, for each $i\in[1,k]$, at least one vertex among $\{\vx{i},\vx{i}',\nvx{i},\nvx{i}'\}$.
    \end{lemma}
    \begin{proof}
Snapshot $1$ is the only one where $c'$ appears. Since in this snapshot, $G=(V, E_1)$ contains exactly one edge, the edge $cc'$, this means that for $c'$ to be "corrupted" afterward, any "temporal zero forcing set" $S$ of $\mathcal{G}$ contains at least one of $c$ or $c'$.
%We can use a similar argument for the vertices 
We can use a slightly less direct version of the same argument applied to the $\vx{i}'$ and the $\nvx{i}'$. Indeed, observe that for any $i\in [1, k]$, the edges incident to either one of $\vx{i}'$ or $\nvx{i}'$ only occur between snapshots $3i-1$ and $3i+1$.
 %Furthermore, no edge incident to any of $\vx{i},\vx{i}',\nvx{i}$, or $\nvx{i}'$ is in any snapshot before snapshot $3i-1$. This means that for $\vx{i}'$ and $\nvx{i}'$ to be "corrupted" by snapshot $3i+2$ (and any later time), at most one of $\vx{i}'$ or $\nvx{i}'$ must have been the only "non-corrupted" neighbor of $c$ at the beginning of snapshot $3i+1$. Assume w.l.o.g. that it is $\vx{i}'$, and that $\vx{i}' \notin S$ (as otherwise we are done). Then by the same argument, $\vx{i}'$ must have been the only "non-corrupted" neighbor of $c$ in a previous snapshot. Since the only previous snapshot where the edge $c\vx{i}'$ appears is snapshot $3i-1$, we conclude that $\vx{i}\in S$. The same argument applies for $\nvx{i}'$. 
 Suppose by contradiction that $S \cap \{\vx{i},\vx{i}',\nvx{i},\nvx{i}'\} = \emptyset$. Then, there is no "corruption" in the snapshots $3i-1, 3i, 3i+1$ and thus $\vx{i}'$ and $\nvx{i}'$ are "not corrupted" at the end of snapshot $3i+1$. This contradicts the fact that these 3 snapshots are the only ones where $\vx{i}'$ and $\nvx{i}'$ appear.
    \end{proof}

    \begin{corollary}\label{col:goodset}
          If $\mathcal{G}$ admits $S$ a "temporal zero forcing set" of size $k+1$, then $S$ contains exactly one vertex in $\{c, c'\}$, and exactly one vertex in each $\{\vx{i},\vx{i}',\nvx{i},\nvx{i}'\}$ for $i\in [1, k]$.
    \end{corollary}
    \begin{proof}
      This is a simple pigeonhole principle applied to Lemma~\ref{cl:3sat-c}.
    \end{proof}
We are now ready to prove Lemma~\ref{lem:np_hard_star}.
\begin{proof}
We first prove the first direction, and assume there exists a valuation $\sigma$ that satisfies $\Phi$. Recall that $k$ is the number of variables in $\Phi$. We build $S$, a "temporal zero forcing set" of $\mathcal(G)$, of size $k+1$ consisting of the vertex $c$, and for every $i\in [1, k]$ if the variable $x_i$ is set to True by $\sigma$, then we add $\vx{i}$, and add $\nvx{i}$ otherwise. We prove that at the end of snapshot $4k + 6m + 1$, every vertex of $\mathcal(G)$ has been successfully "corrupted".
There are several cases:

\begin{itemize}
\item the vertex $c'$ gets "corrupted" in snapshot $1$ since $c\in S$.
\item For every $i\in[1,k]$, both $\vx{i}'$ and $\nvx{i}'$ are "corrupted" at the end of snapshot $3i+1$. Indeed, note that at least one of them is "corrupted" at the beginning of snapshot $3i+1$, as either $\vx{i}$ or $\nvx{i}$ is in $S$. 
\item For all $i\in[1,m]$, we claim that at least one of $\C{i,1}$, $\C{i,2},\C{i,3}$ is "corrupted" at the end of snapshot $3k+1+6i$, as one of the three variables (or its negation) in $C_i$, say $l_{i, 1}$, is set to True (or False for negations), hence its corresponding vertex in $\mathcal{G}$ is in $S$, which thus "corrupts" the vertex $\C{i, 1}$. We are now left to prove that $\C{i, 2}$ and $\C{i, 3}$ are indeed "corrupted" by the end of snapshot $3k+3+6i$. Observe that snapshots $3k+1+6i$ to $3k+3+6i$ enumerate the pairs in $\{\C{i, 1}, \C{i, 2},\C{i, 3}\}$, which concludes the case.
\item Finally, for $i\in[1,k]$, it is direct to observe that both $\vx{i}$ and $\nvx{i}$ are "corrupted" at the end of snapshot $3k + 6m + 4$, as either $\vx{i}$ or $\nvx{i}$ is in $S$.
\end{itemize}
Hence $S$ is a "temporal zero forcing set" of $\mathcal{G}$.

We now prove the reverse direction. Assume there exists a "temporal zero forcing set" $S$ of $\mathcal(G)$ of size $k+1$. We prove that $\Phi$ must be satisfiable. By Corollary~\ref{col:goodset}, $S$ must contain exactly one vertex in $\{c, c'\}$, and exactly one vertex in each $\{\vx{i},\vx{i}',\nvx{i},\nvx{i}'\}$ for $i\in [1, k]$. We define the valuation $\sigma$ to be the valuation obtained by setting $x_i$ to True if either one of $\vx{i},\vx{i}'$ is in $S$, and False otherwise, and claim that $\sigma$ satisfies $\Phi$. Recall that $S$ is a "temporal zero forcing set", that is, at the end of snapshot $4k + 6m + 4$, all the vertices of $\mathcal(G)$ are "corrupted". In particular, for every $i\in [1, m]$, all three vertices of $\{\C{i, 1}, \C{i, 2},\C{i, 3}\}$ have been "corrupted". By Corollary~\ref{col:goodset}, none of those three vertices is in $S$, which implies that at least one of them was "corrupted" before the end of snapshot $3k+3+6i$. Assume w.l.o.g. that this vertex is $\C{i, 1}$ and let $\x{j}$ (with $j\in [1, k]$) be the first literal of $C_{i}$. The only snapshot in which the edge $c\C{i, 1}$ appears in $\mathcal{G}$ before snapshot $3k+3+6i$ is at snapshot $3k+1+6i$, where the only two edges are $c\C{i, 1}$ and $c\vx{j}$. This means that $\vx{j}$ must have been "corrupted" at this point. With a similar argument, either $\vx{j}\in S$ and we are done, or it must be that $\vx{j}$ was "corrupted" at an earlier step, which thus is at snapshot $3j-1$ by $\vx{j}'$. Either way, $\x{j} = l_{i, 1}$ is True in $\sigma$, so the clause $C_j$ is satisfied, which concludes the proof. 
\end{proof}

We are now following with much stronger results, starting by proving the hardness of \textsc{"Temporal Zero Forcing Set@@problem"} when three parameters are bounded: the "lifespan" of the input "temporal graph", the maximum degree of its "underlying graph", and the maximum degree of every snapshot:

\begin{theorem}\label{th:np-bounded-all}
  \textsc{"Temporal Zero Forcing Set@@problem"} is \texttt{NP}-hard when
  restricted to "temporal graphs" whose "lifespan" is at most $10$ and whose "underlying graph" has maximum degree $6$ and each snapshot has maximum degree $1$.
\end{theorem}

The proof of this result, which uses a similar construction as the one to prove Theorem~\ref{th:np-hard}, is mostly done through the following lemma.

\begin{lemma}\label{lemma:main-np-bounded-all}
    Given a \textsc{$3$-SAT} formula $\Phi$ with $k$ variables where each variable appears at most 4 times, there exists a "temporal graph" $\mathcal{G}$ whose "lifespan" is at most $10$, whose "underlying graph" has maximum degree $6$, in which each snapshot has maximum degree $1$, such that $"TZ"(\mathcal{G})= k$ if and only if $\Phi$ is satisfiable.
\end{lemma}

\begin{proof}
The construction is very similar to the one in the proof of Theorem~\ref{th:np-hard}. In that proof, each snapshot (except the first) consists of two vertices $u$ and $v$ connected to the center $c$, which is already "corrupted". Hence, in this snapshot, both $u$ and $v$ are "corrupted" if and only if $u$ or $v$ was "corrupted" before. We get the same property if, in the snapshot, $uv$ is an edge, and they have no other neighbors. With this observation, we are able to \emph{merge} snapshots, in the sense that if we have consecutive snapshots of the form $\{u_i,v_i\}$ in the construction of Theorem~\ref{th:np-hard}, where each vertex appears at most once, we can make a single snapshot where the set of edges are the corresponding $u_iv_i$. We use a version of \textsc{$3$-SAT} where each variable appears at most 4 times to limit the number of snapshots.\\

    Consider a \textsc{$3$-SAT} formula $\Phi$ in its conjunctive normal form with exactly $3$ literals per clause, where each variable appears at most 4 times. Let $(C_i)_{i\in[1,m]}$ be its clauses. For $i\in[1,m]$, let $(l_{i,j})_{j\in \{1,2,3\}}$ be the literals of $C_i$. We have $\displaystyle \Phi=\bigvee_{i=1}^m C_i=\bigvee_{i=1}^m \bigwedge_{j=1}^3 l_{i,j}$. Let $k$ be the number of variables used in the formula $\Phi$. For each $(i,j)\in [1,n]\times \{1,2,3\}$, there exists $l \in [1,k]$ such that $l_{i,j}=\x{l}$ or $l_{i,j}=\nx{l}$.

    Let us construct a temporal star $\mathcal{G}$ such that $"TZ"(\mathcal{G})= k$ if and only if $\Phi$ is satisfiable. The set of vertices $V$ consists in:

\begin{itemize}
\item For each variable $\x{i}$ (with $i \in [1, k]$), we have four vertices $\vx{i}$, $\nvx{i}$, $\vx{i}'$, and $\nvx{i}'$;
\item For each clause $C_i$ (with $i \in [1,m]$), we have three vertices $\C{i,1}$, $\C{i,2}$, and $\C{i,3}$.
\end{itemize}

We will now describe each snapshot, which allows us to deduce $E$ and $\lambda$: 
\begin{enumerate}
\item Snapshot 1 - $E_1=\cup_{i\in[1,k]}\{\vx{i}\vx{i}',~\nvx{i}\nvx{i}'\}$: This is equivalent to Snapshots $3i-1$ and $3i$ for $i \in [1,k]$ of the construction for Lemma~\ref{lem:np_hard_star}.
\item Snapshot 2 - $E_1=\cup_{i\in[1,k]}\{\vx{i}',\nvx{i}'\}$: This is equivalent to Snapshots $3i+1$ for $i \in [1,k]$ of the construction for Lemma~\ref{lem:np_hard_star}.
\item Snapshot $2+l$, with $l\in[1,4]$ - For each variable $x_i$, we look at its $l$th occurrence in a clause. If $x_i$ appears positively in some clause $C_{j}$, at position $p$, we have the edge $\vx{i}\C{j,p}$ in $E_{2+l}$. If it appears negatively in that clause at that position, we instead add the edge $\nvx{i}\C{j,p}$ in $E_{2+l}$. Those snapshots are equivalent to Snapshots $3k-1+6i$ to $3k+1+6i$ of the construction for Theorem~\ref{th:np-hard}, except that instead of processing the "corruption" clause by clause, we parallelize according to an ordering of uses of a variable.
\item Snapshots 7, 8 and 9 - , $E_7=\cup_{i\in[1,m]}\{\C{i,1}\C{i,2}\}$, $E_8=\cup_{i\in[1,m]}\{\C{i,1}\C{i,3}\}$ and $E_9=\cup_{i\in[1,m]}\{\C{i,2}\C{i,3}\}$. Those snapshots are equivalent to Snapshots $3k+6i+2$ to $3k+6i+4$ of the construction for Theorem~\ref{th:np-hard}.
\item Snapshot 10 - $E_{10}=\cup_{i\in[1,k]}\{\vx{i}\nvx{i}\}$: This is equivalent to Snapshots $3k+6m+2$ to $4k+6m+1$ of the construction for Theorem~\ref{th:np-hard}.
\end{enumerate}

Note that all vertices have at most 6 neighbors in the "underlying graph".

The proof from here is the same as in Lemma~\ref{lem:np_hard_star}. The NP-hardness follows from the fact that restricting the number of appearances of each variable to at most 4 does not change the hardness (see Theorem 2.3 of~\cite{TOVEY198485}).
\end{proof}

\section{Temporal Graphs of Lifespan 1}\label{sec:lifespan1}

In this section, we consider "temporal graphs" with "lifespan" 1. In particular, it is equivalent to the case of static graphs, as we only consider a single set of edges. % we can restate Theorem~\ref{thm:pt1} in terms of static graphs. 
The following result answers a question of Section 3.2 of~\cite{HOGBEN20121994}, which asks whether there exists a characterization of minimum (static) "zero forcing sets@@static" with "propagation" time equal to $1$.  Borrowing from static graph terminology, we say a "temporal zero forcing set" $S$ of a "temporal graph" $\mathcal{G}=(V,E,\lambda)$ is a ""1-step forcing set"" if there exists a "corrupting function" $\texttt{corr}$ such that $\texttt{corr}(0) = S$ and $\texttt{corr}(1) = V$.

\begin{theorem}\label{thm:pt1}
  \textsc{"Temporal Zero Forcing Set@@problem"} is \texttt{NP}-hard when
  restricted to "temporal graphs" whose "lifespan" is $1$.
\end{theorem}

\begin{proof}
  Let $\Phi = C_1 \land C_2 \land \dots \land C_r$ be a \textsc{$3$-SAT} formula of $r > 2$ clauses over $\ell > 2$ variables $x_1, x_2, \dots, x_\ell$. We construct a graph $G$ (which is equivalent to a "temporal graph" of "lifespan" 1) as follows:
    \begin{itemize}
      \item Fix a constant $Z \ge 3r + 3$. For each variable $x_i$, create a clique $K_Z(x_i)$ of size $Z$ consisting of vertices named $x_i^1, x_i^2, \dots, x_i^Z$, and a clique $K_Z(\bar{x}_i)$ of size $Z$ consisting of vertices named $\bar{x}_i^1, \bar{x}_i^2, \dots, \bar{x}_i^Z$.       Connect the cliques $K_Z(x_i)$ and $K_Z(\bar{x}_i)$ by edges $\{x_i^j, \bar{x}_i^j\}$ for $j = 1, 2, \dots, Z$. The union of these two cliques is called the gadget for variable $x_i$ and denoted $K_Z(x_i, \bar{x}_i)$. See Figure~\ref{fig:gadget_layer1} for an illustration of the gadget $K_Z(x_i, \bar{x}_i)$.
      \item For each clause $C_j = (l_{j,1} \lor l_{j,2} \lor l_{j,3})$, create a vertex $c_j$. Also create three vertices $u_{j,1}$, $u_{j,2}$, and $u_{j,3}$ and connect each of them to $c_j$. The vertex $u_{j,k}$ is also connected to all vertices of the clique $K_Z(l_{j,k})$ corresponding to the literal $l_{j,k}$.
      \item Connect all vertices $u_{j,k}$ to each other to form a clique of size $3r$.
      \item Connect all vertices $c_j$ to each other to form a clique of size $r$.
    \end{itemize}

    We will refer to the cliques $K_Z(x_i)$ and $K_Z(\bar{x}_i)$ as layer 1, the vertices $u_{j,k}$ as layer 2, and the vertices $c_j$ as layer 3. In particular, there are no edges between layers 1 and 3. Let $G$ denote the graph constructed this way. See Figure~\ref{fig:graph_TZFS_fusion} for an illustration of the graph $G$.

    \begin{figure}[t]
      \centering
      \includegraphics[width=0.4\textwidth]{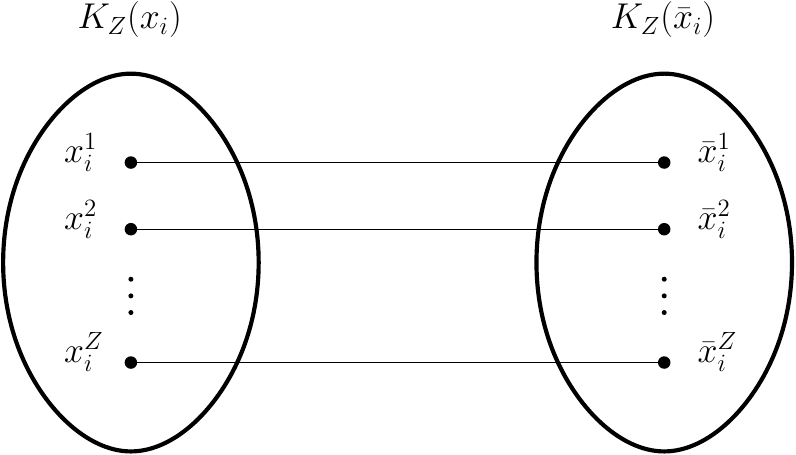}
      \caption{Illustration of the gadget $K_Z(x_i, \bar{x}_i)$ for a variable $x_i$.}
      \label{fig:gadget_layer1}
    \end{figure}

    % \begin{figure}
    %   \centering
    %   \includegraphics[width=0.8\textwidth]{figures/branching_clause.pdf}
    %   \caption{Illustration of edges related to a clause $C_j = x_3 \wedge \bar x_7 \wedge x_8$.}
    %   \label{fig:branching_clause}
    % \end{figure}

    % \begin{figure}
    %   \centering
    %   \includegraphics[width=0.8\textwidth]{figures/graph_TZFS_NPhardness.pdf}
    %   \caption{Illustration of the graph $G$ constructed from a \textsc{$3$-SAT} formula $\Phi$ with $r$ clauses and $\ell$ variables.}
    %   \label{fig:graph_TZFS_NPhardness}
    % \end{figure}

    \begin{figure}[t]
      \centering
      \includegraphics[width=\textwidth]{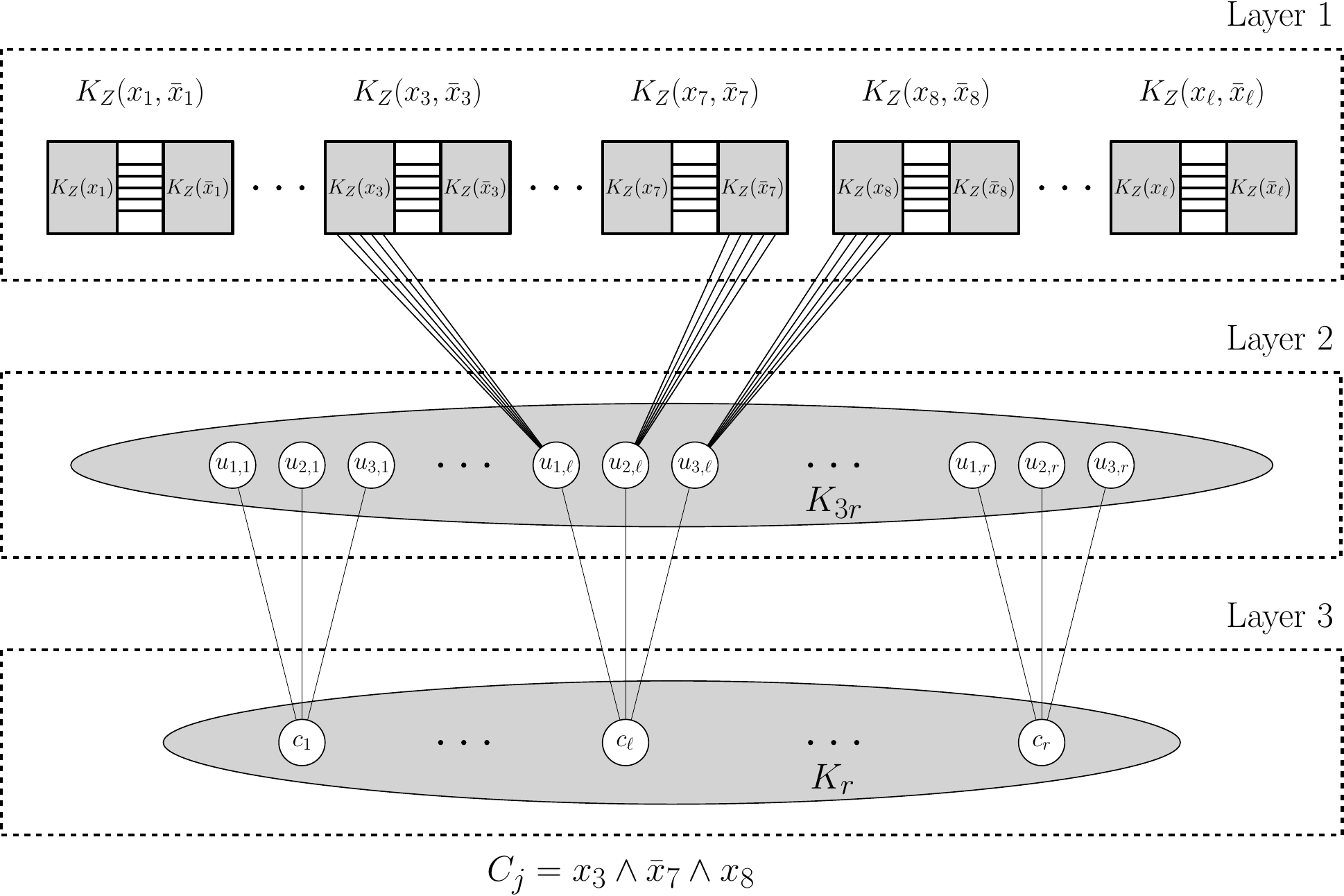}
      \caption{Illustration of the graph $G$ constructed from a \textsc{$3$-SAT} formula $\Phi$ with $r$ clauses and $\ell$ variables. We detail the edges related to the clause $C_j = x_3 \wedge \bar x_7 \wedge x_8$.}
      \label{fig:graph_TZFS_fusion}
    \end{figure}

    We show that $\Phi$ is satisfiable if and only if $Z(G) \le Z \times \ell + 3r$.

  \begin{lemma}\label{claim:NP1stepClaim1}
    If $\Phi$ is satisfiable, then $Z(G) \le Z \times \ell + 3r$.
  \end{lemma}

  \begin{proof} 
    Assume that $\Phi$ is satisfiable and let $\sigma$ be an assignment that satisfies $\Phi$. We construct a set of vertices $S$ as follows:
    \begin{itemize}
      \item For each variable $x_i$, if $\sigma(x_i) = \text{true}$, we add all vertices of the clique $K_Z(x_i)$ to $S$. Otherwise, we add all vertices of the clique $K_Z(\bar{x}_i)$ to $S$.
      \item We also add all the $3r$ vertices of layer 2 to $S$.
    \end{itemize}

    The set $S$ thus constructed is a "1-step forcing set" for the graph $G$. Indeed, for any gadget $K_Z(x_i, \bar{x}_i)$, all vertices of one of the two cliques are in $S$, as are all vertices of layer 2. Thus, if a vertex $x_i^j$ is not in $S$, it is the only "non-corrupted" neighbor of $\bar{x}_i^j$ and will thus be "corrupted" by this vertex. Similarly, if a vertex $\bar{x}_i^j$ is not in $S$, it will be "corrupted" by $x_i^j$. Thus, all vertices of layer 1 will be "corrupted" after one step.
    Since all vertices in layer 2 are in $S$, we only need to corrupt the vertices in layer 3. For each clause $C_j$, at least one of the literals $l_{j,k}$ is true under the assignment $\sigma$, so the corresponding vertex $u_{j,k}$ is in $S$ as well as all vertices of the clique $K_Z(l_{j,k})$. Thus, $c_j$ is the only "non-corrupted" neighbor of $u_{j,k}$ and will be "corrupted" by this vertex. Thus, all vertices of layer 3 will be "corrupted" after one step, and thus $S$ is a "1-step forcing set" for the graph $G$. Moreover, the size of $S$ is $Z \times \ell + 3r$, which concludes the proof.

    For pedagogical reasons, we propose a second "1-step forcing set" $S'$ of size $Z \times \ell + 3r$:
    \begin{itemize}
      \item For each variable $x_i$, if $\sigma(x_i) = \text{false}$, we add all vertices of the clique $K_Z(x_i)$ to $S'$. Otherwise, we add all vertices of the clique $K_Z(\bar{x}_i)$ to $S'$.
      \item We add all the $r$ vertices of layer 3 to $S'$.
      \item For each clause $C_j$, we choose a literal $l_{j}$ present in $C_j$ that is true under the assignment $\sigma$. We add to $S'$ all vertices of layer 2 except the vertices $u_j = u_{j,k}$ (where $k$ is the index of the literal $l_{j}$ in the clause $C_j$) corresponding to the literals chosen for each clause.
    \end{itemize}

    The set $S'$ thus constructed is also a "1-step forcing set" for the graph $G$. Indeed, the only "non-corrupted" neighbor of each vertex $c_j$ is $u_{j}$, so it becomes "corrupted" after one step. Thus, all vertices of layer 2 will be "corrupted" after one step. Note now that the neighbors of the vertices $u_{j}$ lying in layer 1 are all "non-corrupted", since $l_{j}$ is true under the assignment $\sigma$. Then if a vertex $x_i^j$ is not in $S'$, it is the only "non-corrupted" neighbor of $\bar{x}_i^j$ and will thus be "corrupted" by this vertex. Similarly, if a vertex $\bar{x}_i^j$ is not in $S'$, it will be "corrupted" by $x_i^j$. Thus, all vertices of layer 1 will be "corrupted" after one step and thus $S'$ is a "1-step forcing set" for the graph $G$. Moreover, the size of $S'$ is also $Z \times \ell + 3r$, which concludes the proof.
  \end{proof}

  We now aim to prove the converse, that is, if $Z(G) \le Z \times \ell + 3r$, then $\Phi$ is satisfiable. Let $S$ be a "1-step forcing set" for the graph $G$ of size at most $Z \times \ell + 3r$.

  \begin{lemma}\label{claim:NP1stepClaim2}
    For every gadget $K_Z(x_i, \bar{x}_i)$, the set $S$ contains all vertices of one of the two cliques $K_Z(x_i)$ or $K_Z(\bar{x}_i)$.
  \end{lemma}

  \begin{proof}
    Suppose there exists a gadget $K_Z(x_i, \bar{x}_i)$ such that $S$ contains neither all vertices of the clique $K_Z(x_i)$ nor all vertices of the clique $K_Z(\bar{x}_i)$. Let $x_i^j$ be a vertex of the clique $K_Z(x_i)$ not in $S$ and $\bar{x}_i^{j'}$ a vertex of the clique $K_Z(\bar{x}_i)$ not in $S$. The vertex $x_i^j$ cannot be "corrupted" by $\bar{x}_i^j$ because $\bar{x}_i^{j'}$ is "non-corrupted". The vertex $x_i^j$ is "corrupted" either by a vertex of the clique $K_Z(x_i)$ or by a vertex of layer 2. In both cases, this requires that all vertices of the clique $K_Z(x_i)$ except $x_i^j$ are in $S$. The same holds for $\bar{x}_i^{j'}$. Hence there are at least $2Z-2$ vertices of the gadget $K_Z(x_i, \bar{x}_i)$ in $S$. Since $Z \ge 3r + \ell + 1$, there are at least $3r + \ell$ vertices of the gadget $K_Z(x_i, \bar{x}_i)$ in $S$.

    Thus, for each $i$, either one of the two cliques $K_Z(x_i)$ or $K_Z(\bar{x}_i)$ is contained in $S$ (so there are at least $Z$ vertices of $K_Z(x_i, \bar{x}_i)$ in $S$), or this is not the case, and there are at least $2Z-2$ vertices of $K_Z(x_i, \bar{x}_i)$ in $S$. If we are in the second case for at least one gadget, then there are at least $2Z-2 + (\ell-1)Z = (\ell+1)Z -2$ vertices in $S$. But $(\ell+1)Z -2 > \ell Z + 3r$ since $Z \ge 3r + 3$. Therefore, if the second case occurs for any gadget, then $|S| > \ell Z + 3r$, a contradiction. Hence for every gadget $K_Z(x_i, \bar{x}_i)$, the set $S$ contains all vertices of one of the two cliques $K_Z(x_i)$ or $K_Z(\bar{x}_i)$.
  \end{proof}

  \begin{lemma}\label{claim:NP1stepClaim3}
    We can reduce to the case where, for each gadget $K_Z(x_i, \bar{x}_i)$, one of the two cliques $K_Z(x_i)$ or $K_Z(\bar{x}_i)$ is contained in $S$, and the other one is completely "non-corrupted".
  \end{lemma}

  \begin{proof}
    Suppose that $S$ contains more than $Z$ vertices of a gadget $K_Z(x_i, \bar{x}_i)$. By Claim~\ref{claim:NP1stepClaim2}, we know that $S$ contains all vertices of one of the two cliques $K_Z(x_i)$ or $K_Z(\bar{x}_i)$. Without loss of generality, suppose that $S$ contains all vertices of the clique $K_Z(x_i)$. Let $\bar{x}_i^j$ be a vertex of the clique $K_Z(\bar{x}_i)$ that is in $S$. Note that there are at least three "non-corrupted" vertices in the clique $K_Z(\bar{x}_i)$ because otherwise we would have at least $2Z-2$ vertices of $K_Z(x_i, \bar{x}_i)$ in $S$, which is impossible by the same counting argument at the end of the proof of Claim~\ref{claim:NP1stepClaim2}. Thus, the only neighbor of $\bar{x}_i^j$ that does not necessarily have at least two "non-corrupted" neighbors is the vertex $x_i^j$. Removing the vertex $\bar{x}_i^j$ from $S$ does not change the "corruption" process except if $x_i^j$ was supposed to corrupt a vertex.

    If $x_i^j$ has no "non-corrupted" neighbors, then we can remove the vertex $\bar{x}_i^j$ from $S$ because $x_i^j$ will corrupt it.

    If $x_i^j$ has a unique "non-corrupted" neighbor $u_{i,k}$, then we can set $S' = S \setminus \{\bar{x}_i^j\} \cup \{u_{i,k}\}$. Thus, $x_i^j$ "corrupts" $\bar{x}_i^j$ instead of $u_{i,k}$, and $S'$ remains a "1-step forcing set" for $G$.

    If $x_i^j$ has at least two "non-corrupted" neighbors in layer 2, then by construction, this is the case for all vertices of the clique $K_Z(x_i)$ and therefore no one can corrupt the "non-corrupted" vertices of the clique $K_Z(\bar{x}_i)$. Indeed, no vertex in layer 2 can corrupt the "non-corrupted" vertices of the clique $K_Z(\bar{x}_i)$ because they form a clique with at least two "non-corrupted" vertices. Thus $S$ is not a "1-step forcing set" for the graph $G$, which is a contradiction.

    We can repeat this process until we have removed all vertices of the clique $K_Z(\bar{x}_i)$ from $S$. Thus, we can reduce to the case where for each gadget $K_Z(x_i, \bar{x}_i)$ one of the two cliques $K_Z(x_i)$ or $K_Z(\bar{x}_i)$ is contained in $S$ and that the other one is completely "non-corrupted".
  \end{proof}

  \begin{lemma}\label{claim:NP1stepClaim4}
    If $S$ contains all vertices of layer 2, then $\Phi$ is satisfiable.
  \end{lemma}

  \begin{proof}
    We define the valuation $\sigma$ such that $\sigma(x_i) = \text{true}$ if and only if all vertices of the clique $K_Z(x_i)$ are in $S$. By Claim~\ref{claim:NP1stepClaim3}, this well defines a valuation.

    If $S$ contains all vertices of layer 2 and exactly $Z$ vertices of each gadget $K_Z(x_i, \bar{x}_i)$, by cardinality, $S$ cannot contain vertices from layer 3. Thus, each vertex $c_j$ of layer 3 must be "corrupted" by a vertex $u_{j,k}$ of layer 2. For $u_{j,k}$ to corrupt vertex $c_j$, all vertices of the clique $K_Z(l_{j,k})$ must be in $S$. This means that the literal $l_{j,k}$ is true under the valuation $\sigma$. Thus, for each clause $C_j$, there exists a literal $l_{j,k}$ that is true under the valuation $\sigma$, which means that $\sigma$ satisfies $\Phi$.
  \end{proof}

  \begin{lemma}\label{claim:NP1stepClaim5}
      If $S$ contains all vertices of layer 3, then $\Phi$ is satisfiable.
  \end{lemma}

  \begin{proof}
      We define the valuation $\sigma$ such that $\sigma(x_i) = \text{false}$ if and only if all vertices of the clique $K_Z(x_i)$ are in $S$. By Claim~\ref{claim:NP1stepClaim3}, this well defines a valuation.

      If $S$ contains all vertices of layer 3 and exactly $Z$ vertices from each gadget $K_Z(x_i, \bar{x}_i)$, then by cardinality, $S$ contains at most $2r$ vertices in layer 2. By Claim~\ref{claim:NP1stepClaim3}, each "corrupted" vertex in layer 1 has at least one "not corrupted" neighbor in layer 1, so there can be no "corruption" from layer 1 to layer 2. Therefore, for the vertices of layer 2 to become "corrupted", the vertices of layer 3 must corrupt the vertices of layer 2. Since there are $r$ vertices in layer 3, there should be exactly $r$ "non-corrupted" vertices in layer 2, and each vertex of layer 3 must corrupt exactly one vertex of layer 2.
    
      Thus, each vertex $c_j$ of layer 3 has a unique "non-corrupted" neighbor $u_{j,k}$ in layer 2. Now observe that the corresponding clique $K_Z(l_{j,k})$ cannot contain vertices in $S$, otherwise, the vertex $u_{j,k}$ would be a "non-corrupted" neighbor of the vertices of $K_Z(l_{j,k})$, and then these vertices could not corrupt the vertices of $K_Z(\bar{l_{j,k}})$.
      Hence, the literal $l_{j,k}$ is true under the valuation $\sigma$. Since this holds for each layer 3 vertex $c_j$, it means that for each clause $C_j$, there is a literal $l_{j,k}$ that is true under $\sigma$, and thus $\sigma$ satisfies $\Phi$.
  \end{proof}

  To complete the proof, suppose we are neither in the configuration of Claim~\ref{claim:NP1stepClaim4} nor in that of Claim~\ref{claim:NP1stepClaim5}. By Claim~\ref{claim:NP1stepClaim3}, all vertices of layer 1 have at least one "non-corrupted" neighbor in layer 1, so we cannot have "corruption" from layer 1 to layer 2. Furthermore, since there is at least one "non-corrupted" vertex in layer 2 and layer 2 vertices form a $3r$-clique, we cannot have "corruption" from layer 2 to layer 1 or to layer 3. Similarly, since there is at least one "non-corrupted" vertex in layer 3 and layer 3 vertices form an $r$-clique, we cannot have "corruption" from layer 3 to layer 2. We deduce that to corrupt a "non-corrupted" vertex of layer 2, it must be "corrupted" by another vertex of layer 2, which implies that all other vertices of layer 2 are in $S$. Similarly, to corrupt a "non-corrupted" vertex of layer 3, all other vertices of layer 3 must be in $S$. This imposes that $S$ contains $Z \times \ell$ vertices in layer 1 (Claim~\ref{claim:NP1stepClaim3}), $3r - 1$ vertices in layer 2, and $r-1$ vertices in layer 3. Thus $|S| \ge Z \times \ell + 4r - 2 > Z\times \ell + 3r$ whenever $r > 2$, which is a contradiction. We conclude that we must be in the configuration of Claim~\ref{claim:NP1stepClaim4} or in that of Claim~\ref{claim:NP1stepClaim5}, and thus $\Phi$ is satisfiable.
\end{proof}

\section{Complexities for Temporal Stars}\label{sec:temporal-stars}
Theorem~\ref{th:np-hard} shows that having a star as the "underlying graph" is enough for "Temporal Zero Forcing Set@@problem" to be NP-hard. We now prove that the problem remains hard, even when parametrized by the size of the zero forcing set. However, we also prove that the problem is polynomial when the number of snapshots is low: if each snapshot must corrupt a vertex, we provide an algorithm. The polynomiality remains even when a logarithmic number of snapshots are useless, in the sense that no vertex gets "corrupted" at their time.

% we investigate which additional restrictions could reduce the complexity of \textsc{"Temporal Zero Forcing Set@@problem"} when the input "temporal graph" is a temporal star. Using the \textsc{$k$-Multicolored Clique} reduction technique (see~\cite{FELLOWS200953}), we prove that \textsc{"Temporal Zero Forcing Set@@problem"} is \texttt{W[1]}-hard when parameterized by the solution size, that is, the "temporal zero forcing number".

\begin{theorem}\label{th:w1}
  \textsc{"Temporal Zero Forcing Set@@problem"} is \texttt{W[1]}-hard parameterized by the "temporal zero forcing number" of the input graph when
  restricted to "temporal graphs" whose "underlying graph" is a star.
\end{theorem}

We prove this result by reducing from the Multicolored Clique problem, which consists of finding a clique of size $k$ in a $k$-colored graph such that each vertex in the clique has a different color. This problem is known to be \texttt{W[1]}-hard when parametrized by $k$; see, for instance, Lemma 1 of~\cite{FELLOWS200953}. 

\probl{""Multicolored Clique@@problem""}
{A graph $\mathcal{G}=(V,E)$, an integer $k$,  and a $k$-coloring $col:V\to [1,k]$.}
{A clique $Cl\subset V$ of size $k$ such that $\forall u,v\in C^2$, $col(u)\neq col(v)$ or a correct output that such a set does not exist.}

We prove that "Temporal Zero Forcing Set@@problem" is \texttt{W[1]}-hard parametrized by the "temporal zero forcing number" by reducing from the Multicolored Clique problem, which consists of finding a clique of size $k$ in a $k$-colored graph such that each vertex in the clique has a different color. This problem is known to be \texttt{W[1]}-hard when parametrized by $k$; see, for instance, Lemma 1 of~\cite{FELLOWS200953}.

\probl{"Multicolored Clique@@problem"}
{A graph $\mathcal{G}=(V,E)$, an integer $k$,  and a $k$-coloring $col:V\to [1,k]$.}
{A clique $Cl\subset V$ of size $k$ such that $\forall u,v\in C^2$, $col(u)\neq col(v)$ or a correct output that such a set does not exist.}

\begin{lemma}
Given a Multicolored Clique instance with $k$ colors, there exists a temporal star $\mathcal{G}$ such that $"TZ"(\mathcal{G})= k+1$ if and only if the Multicolored Clique instance has a colorful clique of size $k$.
\end{lemma}
%\begin{proof}

    Consider a graph $\mathcal{G}=(V,E)$  and a $k$-coloring $col:V\to [1,k]$ of its vertices. The graph has $n=|V|$ vertices and $m=|E|$ edges. For all $i\in[1,k]$ we denote $V_i=\{v\in V: col(v)=i\}$ and $n_i=|V_i|$. Our construction will be similar to the one used in the proof of Theorem~\ref{th:np-hard}.

    Let us construct a temporal star $\mathcal{G}=(V_\mathcal{G},E_\mathcal{G},\lambda)$ such that $"TZ"(\mathcal{G})= k+1$ if and only if $G$ has a multicolored $k$-clique. We start the construction of $\mathcal{G}$ with a single vertex $c$, which we call the \emph{center}. All the other vertices are connected to $c$. They are:
\begin{itemize}
\item A vertex $c'$;
\item For each vertex $v\in V$, we have two vertices $v$ and $v'$ in $V_\mathcal{G}$;
\item For each edge $e\in E$, we have a vertex $x_e$;
\item For each $1\le i< j \le k$, we have a vertex $y_{i,j}$.
\end{itemize}

To simplify the description $\lambda$, we introduce the value $s=\sum_{i\in[1,k]}\frac{n_i(n_i-1)}2$. The set of edges $E_\mathcal{G}$ corresponds to $\{cv\}_{v\in V_\mathcal{G}\setminus\{c\}}$. To describe $\lambda$, we will actually describe each snapshot. As all edges are connected to $c$, we will describe a snapshot by giving the set of vertices whose edges are present at that time, as well as some intuition behind the choice of this set:

\begin{enumerate}
\item Snapshot 1 - $\{c'\}$: The first snapshot is the only one in which $c'$ appears.
To ensure that $c'$ is "corrupted", either $c'$ or $c$ must be "corrupted" initially. In that case, $c$ is "corrupted" from snapshot $2$ on (otherwise, the whole graph cannot be "corrupted").
\item Snapshots 2 to $n+1$ - For each $v\in V$, we have the set $\{v,v'\}$: both vertices are "corrupted" after this step if and only if at least one of them was initially "corrupted".
\item Snapshots $n+2$ to $n+1+\sum_{i\in[1,k]}\frac{n_i(n_i-1)}2=n+1+s$ - For each $i\in[1,k]$, we use $\frac{n_i(n_i-1)}2$ consecutive snapshots. For each pair of vertices $u,v\in V_i^2$, we have the set $\{u',v'\}$ (one after another). For each $i$, all vertices in $V_i$ are "corrupted" at the end of those snapshots if and only if some vertex $v\in V_i$ was "corrupted" before time $n+2$.
\item  Snapshots $n+2+s$ to $n+1+s+m$ - For each $uv\in E$, we have the set $\{u,v,x_{uv}\}$. If some $x_{uv}$ was "not corrupted" before those snapshots, it becomes "corrupted" after those snapshots if and only if both $u$ and $v$ were "corrupted" before.
\item  Snapshots $n+2+s+m$ to $n+1+s+2m$ - For each $uv\in E$, we have the set $\{x_{uv},y_{i,j}\}$, where $i=col(u)$ and $j=col(v)$. If some $y_{i,j}$ was "not corrupted" before those snapshots, it becomes "corrupted" after those snapshots if and only if, for some edge $uv$ such that $col(u)=i$ and $col(v)=j$, $x_{uv}$ was "corrupted" before.
\item  Snapshots $n+2+s+2m$ to $n+1+s+3m$ - For each edge $e\in E$, we have the set $\{x_{e}\}$. It ensures that all vertices of the form $x_e$ get "corrupted" if $c$ was already "corrupted".
\item Snapshots $n+2+s+3m$ to $2n+1+s+3m$ - For each $v\in V$, we have the set $\{v\}$. It ensures, as in the previous step, that all vertices of the form $v$ get "corrupted" if $c$ was already "corrupted".
\end{enumerate}

    \begin{lemma}\label{cl:size-multicolor}
        Any corrupting set of size $k+1$ contains either $c$ or $c'$ and, for each $i\in[1,k]$, exactly one vertex of the form $v$ or $v'$, with $v\in V_i$.
    \end{lemma}
    \begin{proof}
The first snapshot is the only one where $c'$ appears. For $c'$ to be "corrupted" at the end, $c$ or $c'$ must be initially "corrupted".

For each vertex $v\in V_i$, $v'$ appears once with $v$, then again for each other $u'$ with $u\in V_i$. If it was "not corrupted" before, it must be "corrupted" in one of those snapshots, since it will never appear again. It implies that at least one vertex in $\bigcup_{v\in V_i}\{v,v'\}$ is initially "corrupted".

As we consider a set of size $k+1$, it means that for each $i\in[1,k]$, exactly one vertex of the form $v$ or $v'$, with $v\in V_i$, is in the initial set.
   \end{proof}

    \begin{claim}
        If there exists a multicolored clique of size $k$ of $G$, then $"TZ"(\mathcal{G})= k+1$.
    \end{claim}
    \begin{proof}
        Suppose $G$ has a multicolored clique of size $k$. We will call, for each $i\in[1,k]$, $v_i$ the vertex of $V_i$ in the clique.

        Consider the set $S$ composed of $c$, and, for each $i\in[1,k]$, $v_i$. We have $|S|=k+1$. We will prove that $S$ is a "temporal zero forcing set" of $\mathcal{G}$.

\begin{itemize}
\item Step 1 ensures that $c'$ gets "corrupted" as $c$ is in $S$.
\item Step 2 ensures that, for $i\in[1,k]$, $v_i'$ gets "corrupted".
\item Step 3 ensures that, for all $i\in[1,k]$, for each $v\in V_i$, $v'$ gets "corrupted".
\item Step 4 ensures that, for all $e\in E_{Cl}$ (the set of edges of the multicolored clique), for each $x_e$ gets "corrupted".
\item Step 5 ensures that, for all $1\le i<j\le k$, $y_{i,j}$ gets "corrupted" by the vertex $x_{v_iv_j}$.
\item Step 6 ensures that $c$ corrupts all remaining vertices of the form $x_e$, with $e\in E$.
\item Step 7 ensures that $c$ "corrupts" all remaining vertices of the form $v$, with $v\in V$.
\end{itemize}
This allows us to conclude that $S$ is a "temporal zero forcing set" of $\mathcal{G}$.
    \end{proof}

    \begin{claim}
        If $"TZ"(\mathcal{G})= k+1$, then there exists a multicolored clique of size $k$ of $G$.
    \end{claim}
    \begin{proof}
        Suppose that $"TZ"(\mathcal{G})= k+1$. Hence, there exists a "temporal zero forcing set" $S$ of $\mathcal{G}$.

        By Claim~\ref{cl:size-multicolor}, either $c$ or $c'$ belongs to $S$, as well as, for each $i\in[1,k]$, some vertex $v_i$ or $v_i'$, with $v_i\in V_i$.
        Consider the set of vertices  $$Cl=\{v_i: i\in[1,k] \text{ and } (v_i\in S \text{ or } v_i'\in S)\}$$

We will prove that $Cl$ is a multicolored clique of $G$. First, by construction, each vertex in $Cl$ has a distinct color. We now prove that for all $u,v\in Cl^2$, $uv\in E$.

Note that before step 4, only $c$, $c'$ and vertices of the form $v$ or $v'$, with $v\in Cl$, are "corrupted". In step 5, a vertex $x_{uv}$ gets "corrupted" if and only if $u$ and $v$ were "corrupted", which means that they are in $Cl$. Let us consider some $y_{i,j}$. As this vertex is not in $S$, it must have been "corrupted" during step 5. As each snapshot in which $y_{i,j}$ appears is of the form $\{x_{uv},x_{i,j}\}$ with $u\in V_i$ and $v \in V_j$, the snapshot that "corrupted" $y_{i,j}$ can only happen if $u=v_i$ and $v=v_j$. This implies that $v_iv_j$ is an edge of $G$.

Hence, $Cl$ forms a multicolored clique of $G$.
    \end{proof}

We now consider a specific case: $k=n-T$ on "temporal graphs" with a star as the "underlying graph". Note that with a star as the "underlying graph", at most one vertex can be "corrupted" in each snapshot (either the center gets "corrupted", or it "corrupts" its only "non-corrupted" neighbor). In this scenario, each snapshot must corrupt exactly one vertex. This case can be solved with a polynomial-time algorithm. For a larger "lifespan", it suffices to guess which snapshots can be ignored for "corruption". Hence, if $T-n+k$ is polylogarithmic, we get a polynomial number of guesses, as stated in Corollary~\ref{cor:log-star}.

\begin{theorem}\label{th:opt-star}
  \textsc{"Temporal Zero Forcing Set@@problem"} can be solved in polynomial time when $k = n - T$, where $T$ is the "lifespan" of the input "temporal graph" and the "underlying graph" of the input "temporal graph" is a star.
\end{theorem}

\begin{proof}
The key element relies on the fact that, if the "underlying graph" is a star, at most one vertex can get "corrupted" in each snapshot. Either the center (called $c$) gets "corrupted" by one of its neighbors, or the center "corrupts" one of its neighbors. The latter case can only happen if all the other vertices are already "corrupted".
If $k = n - T$, one vertex must get "corrupted" in each snapshot. In particular, after the first snapshot, the center vertex must be "corrupted", and at the beginning of all snapshots after the first one, all vertices connected to the center but one must be "corrupted".

Let $V_i$, for $i\in[1,T]$, be the set of vertices connected to the center in snapshot $i$. We will consider two cases:
\begin{itemize}
\item The center $c$ is in $S$. This implies that, for all $i\in[1,T]$; after snapshot $i$, $V_i$ is fully "corrupted". In particular, in snapshot $i$, all vertices but one in $W_i=V_i\setminus\bigcup_{j<i}V_j$ must be in $S$; otherwise, no "corruption" could happen in snapshot $i$. On the other hand, if $W_i=\emptyset$, no "corruption" can happen in that snapshot, which implies that no zero forcing set of size $n-T$ exists.

We deduce that a "temporal zero forcing set" which contains $c$ exists if and only if, for all $i\in[1,T]$, $W_i\neq\emptyset$. This is decidable in polynomial time. Moreover, the construction is straightforward: choose $c$, and for each $W_i$, select all its vertices but one.
\item A neighbor of $c$ initially "corrupts" $c$. For each vertex $v\in V_1$, we consider the case where $v$ "corrupts" $c$ in snapshot 1. We define then, for each $i\ge2$, $W_i=V_i\setminus\left(\{u\}\bigcup_{j<i}V_j\right)$, and do the same reasoning as above.

A "temporal zero forcing set" that "corrupts" $c$ in the first snapshot exists if and only if, for some $v\in V_1$, the corresponding sets $W_i$, for all $i\ge2$, are non-empty.
\end{itemize}
Testing both scenarios can be done in polynomial time, which concludes the proof.
\end{proof}

This result allows us to have an algorithm when $k < n-T$, by guessing the snapshots in which a "corruption" happens:

\begin{corollary}
  \textsc{"Temporal Zero Forcing Set@@problem"} can be solved in time $O\left(\binom{T}{k}F(n,T)\right)$ when the "underlying graph" of the input "temporal graph" is a star, where $F(n,T)$ is the complexity to solve the problem when $k=n-T$.
\end{corollary}
\begin{proof}
The idea is to pick each possible set of $k$ snapshots and run the best algorithm from Theorem~\ref{th:opt-star}. The complexity comes from the fact that we have $\binom{T}{k}$ possible sets.
\end{proof}

 This algorithm can then be used for the case where $T-k$ is polylogarithmic in the size of the input "temporal graph":

\begin{corollary}\label{cor:log-star}
  \textsc{"Temporal Zero Forcing Set@@problem"} can be solved in polynomial time on a "temporal graph" $\mathcal{G}$ when we have $T-k$ to be polylogarithmic in $|V(\mathcal{G})|$ and 
 the  "underlying graph" of $\mathcal{G}$ is a star.
\end{corollary}

\section{Bounded-Degree Trees}\label{sec:trees}

As seen in previous sections, the problem is NP-hard on stars, which are trees of unbounded degree, since their center has degree $n-1$. However, we now provide a polynomial algorithm when the maximal degree of the graph is bounded. We do this using dynamic programming on the tree, from leaves to the root. The bound on the degree allows us to store, for each vertex, a bounded amount of information, limited by the number of children of the vertex, given some arbitrary rooting of the tree.

\begin{theorem}\label{th:trees}
  \textsc{"Temporal Zero Forcing Set@@problem"} can be solved in polynomial time when the "underlying graph" of the input "temporal graph" is a bounded degree tree.
\end{theorem}
\begin{proof}
  Let $d$ be an integer.
  Let $\mathcal{G} = (V,E,\lambda)$ be a "temporal graph" such that $R = (V,E)$ is a tree of maximum degree at most $d$%\mikael{Veut-on mettre $\Delta$ pour le degré max ?}.

  We choose a vertex $ r \ in V$ arbitrarily and root the tree $R$ at $r$.
  Given a vertex $v \in V$, we denote by $R_v$ the subtree of $R$ rooted in $v$,
  by $\texttt{child}_R(v)$ the set of every child of $v$ in $R$, and
  by $\texttt{parent}_R(v)$ the set containing only the parent of $v$ if $v$ is not the root, and an empty set if $v$ is the root.
  In the following, we present a dynamic programming algorithm that processes the tree $R$ in a bottom-up fashion.

  \paragraph{Definition of the dynamic programming table.}
  We first describe the dynamic programming table we need to store for each vertex $v \in V$.
  The main idea is that, for each vertex in the closed neighborhood of $v$, we store in $\sigma_v$ the ``price'' of a "corruption" scenario, which is a guess describing when each vertex in the neighborhood of $v$ has been "corrupted" and by which other vertex (itself if it was "corrupted" at the initial time). To each of these scenarios, we associate through $\sigma_v$ its price, that is, the minimum size of a partial "temporal zero forcing set" matching the prediction.% We also store in $R_v$ how many vertices were already "corrupted" at time $0$. 

  Formally, we define, for each $v \in V$, the function
  $$ \sigma_v : (N_R[v] \to [0,T] \times V) \to \mathbb{N} \cup \{\bot\}. $$
  This function signature associates an integer (or the value $\bot$) with an entry of $(N_R[v] \to [0,T] \times V)$. Such an entry can be understood as a record of "corruption" on a local scale, guessing for each vertex in the (underlying) neighborhood of $v$ a time for their "corruption" in $[0, T]$. To each such $\entry$, we make so that $\sigma_v$ returns the minimum size of a "temporal zero forcing set" matching the prescribed corruptions of $\entry$. 
  Formally, we define $\sigma_v$ as the function such that for each element $\entry$ of $N_R[v] \to [0,T] \times V$, the value $\sigma_v(\entry)$ is the minimum value $k$ such that
  there exists a function $\texttt{pcorr}: V(R_v) \cup \texttt{parent}(v) \to ([0,T] \times V)$, called a ""witness"" of $\entry$, respecting that:
  \begin{enumerate}[start=1,label={(\Alph*)}]
  \item \label{enum:pcorr1} for each $w \in N_R[v]$, $\texttt{pcorr}(w) = \entry(w)$,
  \item \label{enum:pcorr2} for each $w \in V(R_v) \cup \texttt{parent}(v)$ with $\texttt{pcorr}(w) = (t_w, x_w)$, either
    \begin{enumerate}[start=1,label={(\roman*)}]
    \item $t_w = 0$,
    \item $x_w \in V(R_v) \cap N_{G_{t_w}}(w)$ such that for each $w' \in N_{G_{t_w}}[x_w] \setminus \{w\}$, if we denote $(t_{w'}, x_{w'}) = \texttt{pcorr}(w')$, we have $t_{w'} < t_w$, or
    \item $x_w \in V \setminus V(R_v)$.
    \end{enumerate}
  \item \label{enum:pcorr3} $k = | \texttt{pcorr}^{-1}(0) \cap V(R_v)|$

  \end{enumerate}
  or $\bot$ if such a function $\texttt{pcorr}$ does not exist.
  %We call $\texttt{pcorr}$ a ""partial witness"" if it respects only \ref{enum:pcorr1} and \ref{enum:pcorr2}.

  Roughly speaking, there exists a partial "corruption" of the vertices of $V(R_v) \cup \texttt{parent}_R(v)$ compatible with the timing and corrupting vertices provided by $\entry$. The compatibility is checked through the "corruption rule" where a vertex is either in the "temporal zero forcing set" (and $t_w = 0$, which is case (B.i)), the only "non-corrupted" neighbor of a "corrupted" vertex in some snapshot (case (B.ii)), or outside of the subtree for which the prescription $\sigma_v$ is defined (case (B.iii)).

  Let $k_{\texttt{sol}} = \min_{f \in  N_R[r] \to [0,T] \times V} \sigma_r(f)$, where $\bot$ is considered bigger than any element of $\mathbb{N}$. 
  By definition of $\sigma_r$, $k_{\texttt{sol}}$ is the size of the minimum "temporal zero forcing set" of $\mathcal{G}$.
  This can be seen by taking for instance the function $\texttt{corr}: [0,T] \to 2^V$ such that for each $i \in [0,T]$, $\texttt{corr}(i) = \texttt{pcorr}^{-1}(i)$.

  \paragraph{Construction of the dynamic programming table.}
  We explain now how to construct the function $\sigma_v$ for each $v \in V$.
  For each child $c$ of $v$, we assume that $\sigma_{c}$ is correctly constructed.
  For each $\entry: N_R[v] \to [0,T] \times V$, we construct $\sigma_v(\entry)$ to be the minimum value of $k$ such that for each $c \in \texttt{child}_R(v)$, there exists $\entry_c: N_R[c] \to [0,T] \times V$ such that:
  \begin{enumerate}[start=1,label={(\alph*)}]
  \item \label{enum:sigma1} $k_c = \sigma_c(\entry_c)$ is in $\mathbb{N}$ and such that for each $w \in N_R[c] \cap N_R[v]$, $\entry_c(w) = \entry(w)$,
  \item \label{enum:sigma2} for each $z \in \texttt{child}_R(v) \cup \texttt{parent}_R(v)$, if $\entry(z) = (v, t_z)$ for some $t_z \in [1,T]$, then $z \in N_{G_{t_z}}[v]$ and for each $w \in N_{G_{t_z}}[v] \setminus \{z\}$, we have that $\entry(w) = (x_w,t_w)$ with $t_w < t_z$, and
  \item \label{enum:sigma3} $k = 1 + \sum_{c \in \texttt{child}_R(v)}k_c$ if $\entry(v) = (0,x)$, for some $x \in V$, and $k = \sum_{c \in \texttt{child}_R(v)}k_c$ otherwise.
  \end{enumerate}
  or $\bot$ if the conditions cannot be satisfied.

    \paragraph{Correctness of the construction.}
  We now prove the algorithm's correctness, that is $\sigma_r$ can be constructed inductively, and that for each value $\entry$, we have that $\sigma_r(\entry)$ admits a "witness" satisfying \ref{enum:pcorr1}, \ref{enum:pcorr2}, and \ref{enum:pcorr3}.
  We prove it by induction on $R$, and consider $v \in V$. We assume that, for each child $c$ of $v$, the function $\sigma_{c}$ is correctly constructed, and $\sigma_v$ is the function obtained by the construction described in the last paragraph.

  Let $\entry: N_R[v] \to [0,T] \times V$ such that $\sigma_v(\entry)$ is in $\mathbb{N}$, and let $\entry_c$, $c \in \texttt{child}(v)$ be the considered entries of the table of the children of $v$ used to construct $\sigma_v(\entry)$.
  By assumption, for each $c \in \texttt{child}(v)$, there exists a function $\texttt{pcorr}_c: V(R_c) \cup \{v\} \to ([0,T] \times V)$  that is a "witness" of $\entry_c$.
  First note that, combining \ref{enum:pcorr1} for  $\texttt{pcorr}_c$, $c\in \texttt{child}(v)$, with \ref{enum:sigma1} for $\entry$, we have that $\texttt{pcorr}_c(v) = \entry(v)$.
  We can now define $\texttt{pcorr}: V(R_v) \cup \texttt{parent}(v) \to ([0,T] \times V)$ such that for each $w \in V(R_v) \cup~\texttt{parent}(v)$, if $w \in N_R[v]$, we set $\texttt{pcorr}(w) = \entry(w)$, otherwise, there exists $c \in \texttt{child}(v)$ such that $w \in V(R_c)$ and we set $\texttt{pcorr}(w) = \texttt{pcorr}_c(w)$.
  By construction, $\texttt{pcorr}$ respects Condition~\ref{enum:pcorr1}.
  Combining \ref{enum:sigma2} for $\entry$ with \ref{enum:pcorr2} for each $\texttt{pcorr}_c$, $c\in \texttt{child}(v)$, we obtain Condition~\ref{enum:pcorr2} for $\texttt{pcorr}$.
  We assume that $\entry(v) = (t_v,x_v)$ for some $t_v \in [1,T]$ and $x_v \in V$.
  By construction, we have $|\texttt{pcorr}^{-1}(0) \cap V(R_v)| = \sum_{c \in \texttt{child}(v)}|\texttt{pcorr}_c^{-1}(0) \cap V(R_c)| = \sum_{c \in \texttt{child}(v)}\sigma_c(\entry_c) = \sigma_v(\entry)$, and so, Condition~\ref{enum:pcorr3} is satisfied.
  The same argument applies if $\entry(v) = (0,x_v)$ for some $x_v \in V$.

  \paragraph{Runtime analysis.}
  For each vertex $v \in V$, we need to construct $\sigma_v : (N_R[v] \to [0,T] \times V) \to \mathbb{N} \cup \{\bot\}$.
  The size of $N_R[v] \to [0,T] \times V$ is at most $((T+1) \cdot n)^{d+1}$ where $n = |V|$.
  When constructing $\sigma_v(\entry)$ for some $\entry$ of $N_R[v] \to [0,T] \times V$, we have to check for each entry of the table of each child of $v$ whether we can combine them, providing a running time for the computation of $\sigma_v$ in $\big(((T+1) \cdot n)^{d+1}\big)^{d+1}\cdot (d+T)^{O(1)}$ where the $(d+T)^{O(1)}$ appears for the verification that the constructed entry respect Conditions \ref{enum:sigma1}, \ref{enum:sigma2}, and \ref{enum:sigma3}.
  Thus, we obtain an algorithm running in time $((T+1) \cdot n)^{(d+1)^2}\cdot (d+T)^{O(1)} \cdot n$.

  In the presented algorithm, for ease of notation, we keep the possibility that a vertex can be "corrupted" by any vertex.
  By definition, this is not true; only a vertex's neighbors can corrupt it.
  With this in mind, we can prune these entries, and we obtain that, for a vertex $v \in V$, the number of entries of $\sigma_v$ that make sense is not $((T+1) \cdot n)^{d+1}$ but $((T+1) \cdot d)^{d+1}$, providing an algorithm running in time $((T+1) \cdot d)^{(d+1)^2}\cdot (d+T)^{O(1)} \cdot n$.
\end{proof}

This algorithm can be generalized when the maximal degree is a parameter, to get the following result:

\begin{corollary}
  \textsc{"Temporal Zero Forcing Set@@problem"} is fixed-parameter tractable by the "lifespan" plus the maximum degree when the input graph is a tree.
\end{corollary}

\section{Conclusion}\label{sec:ccl}

We have adapted the \textsc{"Zero Forcing Set@@problem"} in "temporal graphs", and shown its complexity. We proved that the problem is \texttt{NP}-hard even when the "underlying graph" is a star, or when the "lifespan", the maximum degree of the "underlying graph", and the maximum degree of every snapshot are bounded. This prompted us to investigate these parameters from different angles. We answered an open question of~\cite{HOGBEN20121994} by showing that computing the "temporal zero forcing number" with "lifespan" $1$ is hard in general. However, we also obtained tractability results on trees of bounded degree. 
Our exploration of \textsc{"Temporal Zero Forcing Set@@problem"} strongly supports our claim that the dynamicity of the network adds a significant layer of hardness to the original "problem@@static". In particular, finding a minimum "temporal zero forcing set" is hard on temporal stars; recall Theorem~\ref{th:np-hard}. One of the main consequences of Theorem~\ref{th:np-bounded-all} and Theorem~\ref{th:w1} is that there really is (mostly) no hope of tractability by the solution size, even under strong hypotheses. Due to Theorem~\ref{th:trees}, as well as the insights from Chapter 4 of~\cite{aazami2008hardness}, we propose the following conjecture.

\begin{conjecture}\label{conj:trees}
  \textsc{"Temporal Zero Forcing Set@@problem"} can be solved in \texttt{FPT}-time parameterized by treewidth and the "lifespan" of the input "temporal graphs".
\end{conjecture}

Up to adapting the terminology, Theorem 4.1.8 of~\cite{aazami2008hardness} implies that Conjecture~\ref{conj:trees} holds if all the snapshots are the same. Additionally, the proof uses a dynamic programming algorithm that we think could be adapted to the "temporal equivalent@@problem". 

In our definition, we have made a single step of "corruption" in each snapshot. Another axis of exploration is allowing several rounds of "corruption" in a single snapshot. What new results can we get in that scenario?

\bibliographystyle{splncs04}
\bibliography{biblio}

\end{document}